\documentclass[11pt]{article}    % onecolumn (standardformat)

\usepackage[margin=0.75in]{geometry} %In Folder
\usepackage{amsmath,amssymb,amsthm}
\newtheorem{theorem}{Theorem}[section]

\newtheorem{remark}{Remark}

\newcommand{\bt}{\beta}

\newcommand{\be}{\begin{equation}}
	\newcommand{\ee}{\end{equation}}
\newcommand{\bea}{\begin{eqnarray}}
	\newcommand{\eea}{\end{eqnarray}}
\newcommand{\no}{\nonumber}
\numberwithin{equation}{section}
\begin{document}
\title{The discrete Painlev\'{e} XXXIV hierarchy arising from the gap probability distributions of Freud random matrix ensembles}
\author{Chao Min\thanks{School of Mathematical Sciences, Huaqiao University, Quanzhou 362021, China; Email: chaomin@hqu.edu.cn}\: and Liwei Wang\thanks{School of Mathematical Sciences, Huaqiao University, Quanzhou 362021, China}}
%\authorrunning{Short form of author list} % if too long for running head
\date{January 16, 2025}
% The correct dates will be entered by the editor
\maketitle
\begin{abstract}
We consider the symmetric gap probability distributions of certain Freud unitary ensembles. This problem is related to the Hankel determinants generated by the Freud weights supported on the complement of a symmetric interval. By using Chen and Ismail's ladder operator approach, we obtain the difference equations satisfied by the recurrence coefficients for the orthogonal polynomials with the discontinuous Freud weights. We find that these equations, with a minor change of variables, are the discrete Painlev\'{e} XXXIV hierarchy proposed by Cresswell and Joshi [{\em J. Phys. A: Math. Gen.} {\bf 32} ({1999}) {655--669}]. This is the first time that the discrete Painlev\'{e} XXXIV hierarchy appears in the study of Random Matrix Theory. We also derive the differential-difference equations for the recurrence coefficients and show the relationship between the logarithmic derivative of the gap probabilities, the nontrivial leading coefficients of the monic orthogonal polynomials and the recurrence coefficients.
\end{abstract}

$\mathbf{Keywords}$: Freud unitary ensembles; Gap probabilities; Orthogonal polynomials; Ladder operators;

Recurrence coefficients; Discrete Painlev\'{e} XXXIV hierarchy.

%Painlev\'{e} V; Coulomb fluid; Asymptotics.
$\mathbf{Mathematics\:\: Subject\:\: Classification\:\: 2020}$: 60B20, 42C05.

\section{Introduction}
Consider a random matrix ensemble on the space of $n\times n$ Hermitian matrices $M$ with the probability distribution
$$
\frac{1}{Z_n}\mathrm{e}^{-\mathrm{Tr}\: \mathrm{v}_{0}(M)}dM,\qquad dM=\prod_{i=1}^{n}dM_{ii}\prod_{1\leq i< j\leq n}d\mathfrak{R}M_{ij}d\mathfrak{I}M_{ij},
$$
where $Z_n$ is the normalization constant. This distribution is invariant under all unitary transformations and induces a probability distribution on the eigenvalues $x_1, x_2, \ldots, x_n$ of $M$ given by \cite{Mehta}
	$$
	p\left(x_1, x_2, \ldots, x_n\right) \prod_{j=1}^n d x_j=\frac{1}{Z_n} \prod_{1 \leq i<j \leq n}(x_j-x_i)^2 \prod_{k=1}^n w_0(x_k) d x_k,\qquad x_1, x_2, \ldots, x_n\in \mathbb{R},
	$$
	where $ w_0(x):=\mathrm{e}^{-\mathrm{v}_{0}(x)}$ is a weight defined on $\mathbb{R}$, and all the moments
	$$
	\mu_j:=\int_{-\infty}^{\infty} x^j w_0(x) d x, \qquad j=0,1,2, \ldots
	$$
	exist.
	
	In this paper, we take $w_0(x)$ to be the Freud weight
	$$
	w_0(x)=\mathrm{e}^{-x^{2m}},\qquad m\in\mathbb{Z}^{+},\; x \in \mathbb{R},
	$$
	which corresponds to the Freud unitary ensemble. When $m=1$, it is the well-known Gaussian unitary ensemble (GUE). The probability that the interval $(- a, a) $ is free of  eigenvalues in the Freud unitary ensemble is given by
	\be\label{pnt1}
	\mathbb{P}(n,a)  =\frac{1}{Z_n} \int_{(-\infty, \infty)^n} \prod_{1 \leq i<j \leq n}(x_j-x_i)^2 \prod_{k=1}^n w_0(x_k)\chi_{(-a,a)^{c}}(x_k) d x_k,
	\ee
where $\raisebox{0.5ex}{$\chi$}_{I}(x)$ is the characteristic function of the interval $I$. The asymptotics of $\mathbb{P}(n,a)$ has been studied in \cite{CKM} under suitable double scaling as $n\rightarrow\infty$ for more general Freud unitary ensemble.

	By using Andr\'{e}ief's or Heine's integration formula \cite[p. 27]{Szego}, the $n$-fold integrals in (\ref{pnt1}) can be expressed as the Hankel determinants
\begin{align}
D_n(a):&=\det\left(\int_{-\infty}^{\infty} x^{i+j} w_0(x)\chi_{(-a,a)^{c}}(x) d x\right)_{i, j=0}^{n-1}\no\\
&=\frac{1}{n!}\int_{(-\infty, \infty)^n} \prod_{1 \leq i<j \leq n}(x_j-x_i)^2 \prod_{k=1}^n w_0(x_k)\chi_{(-a,a)^{c}}(x_k) d x_k.\no
\end{align}
In particular, $Z_n=n!D_n(0)$, and we have
	\be\label{pna}
	\mathbb{P}(n, a) = \frac{D_n(a)}{D_n(0)}.
	\ee

It is a known fact that $D_n(a)$ is equal to the product of the square of $L^2$ norms of the corresponding \textit{monic} orthogonal polynomials (see e.g. \cite[(2.1.6)]{Ismail}), i.e.,
	\be\label{hankel}
	D_n(a)=\prod_{j=0}^{n-1}h_{j}(a),
	\ee
	where $h_j(a)>0$ is defined by
	$$
	h_j(a)\delta_{jk}=\int_{-\infty}^{\infty}P_j(x;a)P_k(x;a)w(x;a)dx,\qquad j, k=0,1,2,\ldots.
	$$
	Here $\delta_{jk}$ denotes the Kronecker delta, and $P_j(x; a)$ are the monic polynomials of degree $j$, orthogonal with respect to the weight
\be\label{weight}
w(x;a):=w_0(x)\chi_{(-a,a)^{c}}(x),\qquad x \in \mathbb{R}.
\ee
	
	Since $w(x;a)$ is even, the orthogonal polynomials satisfy  the three-term recurrence relation of the form
	$$
	xP_{n}(x;a)=P_{n+1}(x;a)+\beta_{n}(a)P_{n-1}(x;a)
	$$
	with the initial conditions
	$$
	P_0(x;a)=1,\qquad \beta_0(a) P_{-1}(x;a)=0.
	$$
	It can be shown that the recurrence coefficient $\bt_n(a)$ has the expressions
	\be\label{be}
	\beta_{n}(a)=\frac{h_{n}(a)}{h_{n-1}(a)}
	\ee
and
	\begin{equation}\label{alpha1}
		\beta_{n}(a)=\mathrm{p}(n,a)-\mathrm{p}(n+1,a),
	\end{equation}
	where $\mathrm{p}(n, a)$ is the nontrivial leading coefficient of the monic polynomial $P_n(x; a)$, i.e.,
	$$
	P_n\left(x; a\right)=x^n+\mathrm{p}(n, a) x^{n-2}+\cdots,
	$$
	with the initial values $\mathrm{p}(0, a)=\mathrm{p}(1, a)=0$.

	Moreover, from (\ref{alpha1}) we have the identity
	\be\label{sum}
	\sum_{j=0}^{n-1}\bt_j(a)=-\mathrm{p}(n,a),
	\ee
and note that $\bt_0(a)=0$.
	The combination of (\ref{be}) and (\ref{hankel}) gives the relation between $\bt_n(a)$ and $D_n(a)$:
	$$
	\bt_n(a)=\frac{D_{n+1}(a)D_{n-1}(a)}{D_n^2(a)}.
	$$
See \cite{Chihara,Ismail,Szego} for more information about orthogonal polynomials.
	
The ladder operator approach developed by Chen and Ismail \cite{Chen1997} has been demonstrated to be very useful to analyze the recurrence coefficients of orthogonal polynomials; see, e.g., \cite{BC2009,Dai,Filipuk,Min2021,VanAssche}. Similarly as in \cite{BC2009} and noting that
$$
\chi_{(-a,a)^{c}}(x)=1+\theta(x-a)-\theta(x+a),
$$
where $\theta(x)$ is the Heaviside step function, i.e., $\theta(x)$ is $1$ for $x>0$ and $0$ otherwise,
 it can be shown that our monic orthogonal polynomials with the weight (\ref{weight}) satisfy the lowering and raising operator equations
	$$
	\left(\frac{d}{dx}+B_{n}(x)\right)P_{n}(x)=\beta_{n}A_{n}(x)P_{n-1}(x),
	$$
	$$
	\left(\frac{d}{dx}-B_{n}(x)-\mathrm{v}_0'(x)\right)P_{n-1}(x)=-A_{n-1}(x)P_{n}(x),
	$$
	where $\mathrm{v}_0(x)=-\ln w_0(x)$, and $A_{n}(x)$ and $B_{n}(x)$ are the functions defined by
	%differentiable and $\mathrm{v}'(z)$ is Lipschitz continuous in the real line, and
	\be\label{an}
	A_{n}(x):=\frac{R_n(a)}{x^2-a^2}+\frac{1}{h_{n}(a)}\int_{-\infty}^{\infty}\frac{\mathrm{v}_0'(x)-\mathrm{v}_0'(y)}{x-y}P_{n}^{2}(y)w(y)dy,
	\ee
	\be\label{bn}
	B_{n}(x):=\frac{x\: r_n(a)}{x^2-a^2}+\frac{1}{h_{n-1}(a)}\int_{-\infty}^{\infty}\frac{\mathrm{v}_0'(x)-\mathrm{v}_0'(y)}{x-y}P_{n}(y)P_{n-1}(y)w(y)dy,
	\ee
	and $R_n(a)$ and $r_n(a)$ are the auxiliary quantities given by
	\be\label{Rnt}
	R_n(a):=\frac{2a w_0(a)}{h_n(a)}P_n^2(a; a),
	\ee
	\be\label{rnt}
	r_n(a):= \frac{2 w_0(a)}{h_{n-1}(a)} P_n(a;a) P_{n-1}(a;a).
	\ee
	We often suppress the $a$-dependence of many quantities such as $\bt_n$ and the auxiliary quantities for brevity, and we believe that this will not lead to any confusion.
	
	It can be proved that the following compatibility conditions for $A_n(x)$ and $B_n(x)$ hold from their definitions in (\ref{an}) and (\ref{bn}):
	%identities for all $z\in\mathbb{C}\cup\{\infty\}$:
	\be
	B_{n+1}(x)+B_{n}(x)=x A_{n}(x)-\mathrm{v}_0'(x), \tag{$S_{1}$}
	\ee
	\be
	1+x(B_{n+1}(x)-B_{n}(x))=\beta_{n+1}A_{n+1}(x)-\beta_{n}A_{n-1}(x), \tag{$S_{2}$}
	\ee
	\be
	B_{n}^{2}(x)+\mathrm{v}_0'(x)B_{n}(x)+\sum_{j=0}^{n-1}A_{j}(x)=\beta_{n}A_{n}(x)A_{n-1}(x). \tag{$S_{2}'$}
	\ee
	Here ($S_{2}'$) is called the sum rule, which is obtained from a suitable combination of ($S_{1}$) and ($S_{2}$).

By using the above compatibility conditions, we obtain some identities for the recurrence coefficients and the auxiliary quantities, from which we derive the difference equations satisfied by the recurrence coefficients of orthogonal polynomials with the weight (\ref{weight}). We find that these equations are closely related to the discrete Painlev\'{e} XXXIV hierarchy \cite{Cresswell}. The relationship between the recurrence coefficients of semi-classical orthogonal polynomials and discrete Painlev\'{e} equations was first observed by Fokas, Its and Kitaev \cite{Fokas}, and see also \cite{Boelen,DM,Fokas2,Magnus,MinJMP,VanAssche} for reference. In addition, the relationship has also been studied based on Sakai's geometric theory of discrete Painlev\'{e} equations in recent years; see, e.g., \cite{Dzhamay,DFS,Hu,Hu2}. We also mention that Yue, Chang and Hu \cite{YCH} proposed a new method called ``stationary reduction method based on nonisospectral deformation of orthogonal polynomials'' for deriving discrete Painlev\'{e}-type equations very recently.

Taking a derivative with respect to $a$ in the equality
	$$
	h_n(a)=\int_{-\infty}^{\infty}P_n^2(x;a)w_0(x)\chi_{(-a,a)^{c}}(x) dx,
	$$
	we find
	\be\label{dhn}
	a\frac{d}{da}\ln h_n(a)=-R_n(a).
	\ee
	It follows from (\ref{be}) that we have the Toda type equation
	\be\label{eq3}
	a\beta_{n}'(a)=\beta_n(a) (R_{n-1}(a)- R_n(a)).
	\ee
This equation can be used to derive the differential-difference equation satisfied by the recurrence coefficient $\bt_n(a)$.

Let $\sigma_n(a)$ be the quantity associated with the logarithmic derivative of the gap probability, i.e.,
$$
\sigma_n(a):=a\frac d{da}\ln\mathbb{P}(n,a).
$$
It can be seen from (\ref{pna}), (\ref{hankel}) and (\ref{dhn}) that
	\be\label{sig}
	\sigma_n(a)=a\frac d{da}\ln D_n(a)=-\sum_{j=0}^{n-1}R_j(a).
	\ee
The identity (\ref{sig}) will play an important role in establishing the relation of $\sigma_n(a)$ to the recurrence coefficient $\bt_n(a)$.

In this paper, we mainly consider the symmetric gap probability (\ref{pnt1}) of the Freud unitary ensemble in the $m=1, 2$ and $3$ cases. Our main result is that the logarithmic derivative of the gap probability can be expressed in terms of the recurrence coefficient $\bt_n$ explicitly, where $\bt_n$ (with a minor change of variables) satisfies the discrete Painlev\'{e} XXXIV hierarchy. It will be seen that our method can be used to study the higher $m$ cases; however, the results will become more and more complicated.
	\section{Gaussian unitary ensemble ($m=1$)}
	We first consider the simplest case for $m=1$. The weight function $w(x)$ now is
	\begin{equation*}\label{wei1}
		w(x;a)=w_0(x)\chi_{(-a,a)^{c}}(x),\qquad x \in \mathbb{R},
	\end{equation*}
	where $w_0(x)$ is the Gaussian weight
	$$
	w_0(x)=\mathrm{e}^{-x^{2}},\qquad x \in \mathbb{R}.
	$$
It is a celebrated discovery of Jimbo, Miwa, M\^{o}ri and Sato (JMMS) \cite{JMMS} that the gap probability of GUE in (\ref{pnt1}) can be expressed in terms of a Painlev\'{e} V transcendent by scaling in the bulk of the spectrum.
This gap probability has also been studied by Cao, Chen and Griffin \cite{cao} using the ladder operator approach. Unfortunately, there were some mistakes in \cite{cao}, which have been corrected in the work of Lyu, Chen and Fan \cite{LCF}. See also \cite{TW,Witte}.

    Following \cite{cao,LCF} and noting that $\mathrm{v_0}(x)=-\ln w(x_0)=x^2$, we find from (\ref{an}) and (\ref{bn}) that
	\be\label{anz1}
	A_n(x)
	=\frac{R_n(a)}{x^2-a^2}+2,
	\ee
	\be\label{bnz1}
	B_n(x) =\frac{x\: r_n(a)}{x^2-a^2},
	\ee
	where $R_n(a)$ and $r_n(a)$ are given by (\ref{Rnt}) and (\ref{rnt}), respectively.
	
	Next we use the compatibility conditions satisfied by $A_n(x)$ and $B_n(x)$ to analyze the recurrence coefficient $\bt_n$ and the auxiliary quantities. Substituting (\ref{anz1}) and (\ref{bnz1}) into ($S_{1}$), we obtain
	\be\label{s1}
	r_{n}+r_{n+1}=R_{n}.
	\ee
	Similarly, substituting (\ref{anz1}) and (\ref{bnz1}) into ($S_{2}'$) gives
	\be\label{s2}
	r_{n}=2\beta_{n}-n,
	\ee
	\be\label{s3}
	a^2 r_{n}^{2}=\beta_{n}R_{n-1}R_{n},
	\ee
	\be\label{s4}
	r_{n}^{2}+2a^2r_{n}+\sum_{j=0}^{n-1}R_{j}=2\beta_{n}(R_{n-1}+R_{n}).
	\ee

Based on the above identities, we have the following results.
	\begin{theorem}\label{th1}
		The recurrence coefficient $\bt_n$ satisfies the second-order nonlinear difference equation
	\begin{equation}\label{de1}
	a^2(2\beta_n-n)^2=\beta_n(2\beta_{n-1}+2\beta_n-2n+1)(2\beta_n+2\beta_{n+1}-2n-1).
	\end{equation}
Let $w_n=\frac{4\beta_{n}}{a^2}$. Then, the difference equation (\ref{de1}) becomes the discrete Painlev\'{e} XXXIV equation \cite[(4.2)]{Cresswell} (see also \cite[(5)]{Joshi})
    	\be\label{dp34}
    	(w_{n+1}+w_n-z_{n+1})(w_n+w_{n-1}-z_n)=\frac{(2w_n-C_3-z_n)(2w_n+C_3-z_{n+1})}{w_n},
    	\ee
    	where $z_n=C_1+C_2 n,\; C_1=-\frac{2}{a^2},\; C_2=\frac{4}{a^{2}}$ and $C_3=\frac{2}{a^{2}}$.
    \end{theorem}
	\begin{proof}
	From (\ref{s1}) and (\ref{s2}) we have
	\begin{equation}\label{s5}
	R_n=2\beta_n+2\beta_{n+1}-2n-1.
	\end{equation}
	Substituting (\ref{s2}) and(\ref{s5}) into  (\ref{s3}), we obtain equation (\ref{de1}). Under the given transformation, we arrive at equation (\ref{dp34}).
	\end{proof}
\begin{theorem}\label{th2}
	The recurrence coefficient $\bt_n$ satisfies the  differential-difference equation
	\begin{equation}\label{dd1}
	a \beta_{n}'(a)=2\beta_{n}(\bt_{n-1}-\bt_{n+1}+1).
	\end{equation}
	\end{theorem}
	\begin{proof}
	Substituting (\ref{s5}) into (\ref{eq3}), we obtain (\ref{dd1}).
	\end{proof}
\begin{remark}
The equation (\ref{dd1}) is similar in form to the Volterra lattice or the Langmuir lattice equation.
\end{remark}
	\begin{theorem}\label{th3}
	The nontrivial leading coefficient $\mathrm{p}(n,a)$ can be expressed in terms of the recurrence coefficient $\bt_n$ as follows:
	\begin{equation}\label{pn1}
	\mathrm{p}(n,a)=-\frac{1}{2}na^2-\beta_n\left(\beta_{n-1}+\beta_n+\beta_{n+1}-n-a^2-\frac{1}{2}\right).
	\end{equation}
    \end{theorem}
	\begin{proof}
	From (\ref{sum}), (\ref{s2}) and (\ref{s1}), we have
\be\label{pn}
\mathrm{p}(n,a)=-\frac{1}{2}\sum_{j=0}^{n-1}r_j-\frac{1}{4}n(n-1)=-\frac{1}{4}\sum_{j=0}^{n-1}R_j+\frac{r_n}{4}-\frac{1}{4}n(n-1).
\ee
Using (\ref{s4}) to eliminate $\sum_{j=0}^{n-1}R_j$ yields
$$
\mathrm{p}(n,a)=\frac{1}{4}\left[r_{n}^{2}+(2a^2+1)r_{n}-2\beta_{n}(R_{n-1}+R_{n})-n(n-1)\right].
$$
Substituting (\ref{s2}) and (\ref{s5}) into the above gives the desired result.
	\end{proof}
	\begin{theorem}
    The quantity $\sigma_n(a)=a\frac d{da}\ln\mathbb{P}(n,a)$ can be expressed in terms of the recurrence coefficient $\bt_n$ as follows:
	\be\label{ex2}
	\sigma_n(a)
	=n(n-2a^2)-4\beta_n(\beta_{n-1}+\beta_n+\beta_{n+1}-n-a^2).
	\ee
	\end{theorem}
	\begin{proof}
	Using (\ref{sig}) and (\ref{pn}) we have
$$
\sigma_n(a)=4\mathrm{p}(n,a)-r_n+n(n-1).
$$
Substituting (\ref{pn1}) and (\ref{s2}) into the above, we arrive at (\ref{ex2}).
	\end{proof}
It was shown in \cite{LCF} that $\sigma_n(a)$ satisfies a nonlinear second-order ordinary differential equation, which can be transformed into the famous JMMS equation (the $\sigma$-form of a particular Painlev\'{e} V equation) under suitable double scaling as $n\rightarrow\infty$.

	\section{Quartic Freud unitary ensemble ($m=2$)}
	We now consider the $m=2$ case and the weight function $w(x)$ reads
	$$
	w(x;a)=w_0(x)\chi_{(-a,a)^{c}}(x),\qquad x \in \mathbb{R},
	$$
	where $w_0(x)$ is the quartic Freud weight $w_0(x)=\mathrm{e}^{-x^{4}}$ and $\mathrm{v_0}(x)=x^4,\;x \in \mathbb{R}$.

    Noting that
	$$
	\frac{\mathrm{v}_0'(x)-\mathrm{v}_0'(y)}{x-y}=4(x^2+xy+y^2)
	$$
	and inserting it into the definitions of $A_n(x)$ and $B_n(x)$ in (\ref{an}) and (\ref{bn}), we find
	\be\label{anz}
		A_n(x) =\frac{R_n(a)}{x^2-a^2}+4 x^2+4\bt_n+4\bt_{n+1},
	\ee
	\be\label{bnz}
		B_n(x) =\frac{x\:r_n(a)}{x^2-a^2}+4x \beta_n,
	\ee
	where use has been made of the three-term recurrence relation to simplify the results, and $R_n(a)$ and $r_n(a)$ are given by (\ref{Rnt}) and (\ref{rnt}), respectively.
	
	Substituting (\ref{anz}) and (\ref{bnz}) into ($S_1$), we obtain
	\begin{equation}\label{s11}
		r_n+r_{n+1}=R_n.
	\end{equation}
	Similarly, substituting (\ref{anz}) and (\ref{bnz}) into ($S_2'$) produces the following four identities:
	\begin{equation}\label{s12}
	a^2 r_n^2=\beta_n R_{n-1} R_n,
	\end{equation}
	\begin{equation}\label{s13}
	r_n= 4\beta_n\left(\bt_{n-1}+\bt_n+\bt_{n+1}\right)-n ,
	\end{equation}
	\be\label{s14}
r_n (a^2+2\bt_n)+\sum_{j=0}^{n-1}(\beta_j+\beta_{j+1})=\beta_n[R_{n-1}+R_n+4(\beta_{n-1}+\beta_n)(\beta_n+\beta_{n+1})],
	\ee
    \begin{align}\label{s16}
    &\:a^2r_n^2+8a^4r_n(a^2+2\bt_n)+2a^2\sum_{j=0}^{n-1}R_j\no\\
    &=\beta_n\big[8a^2R_n(a^2+\beta_{n-1}+\beta_n)+8a^2R_{n-1}(a^2+\beta_n+\beta_{n+1})-R_{n-1}R_n\big].
    \end{align}

    We are now ready to obtain the following results.
	\begin{theorem}\label{the}
	The recurrence coefficient $\bt_n$ satisfies the fourth-order nonlinear difference equation
	\begin{align}\label{de2}
	&a^2\left[4\beta_n(\bt_{n-1}+\bt_n+\bt_{n+1})-n\right]^2\no\\
&=\beta_n\left[4\left(\beta_{n-2}\beta_{n-1}+(\beta_{n-1}+\beta_n)^2+\beta_n\beta_{n+1}\right)-2n+1\right]\no\\
	&\times\left[4\left(\beta_{n-1}\beta_n+(\beta_n+\beta_{n+1})^2+\beta_{n+1}\beta_{n+2}\right)-2n-1\right].
	\end{align}
Let $w_n=\frac{4\beta_{n}}{a^2}$. Then, the difference equation (\ref{de2}) becomes the fourth-order equation in the discrete Painlev\'{e} XXXIV hierarchy \cite[(4.3)]{Cresswell}
    	\begin{align}\label{dp342}
			&w_{n}^{2}(w_{n+2}w_{n+1} + w_{n-1}w_{n-2})(w_{n}-\kappa)+w_{n}^{5}-2w_{n}^{4}(\kappa+2) -w_n(z_nw_{n+2}w_{n+1}+z_{n+1}w_{n-1}w_{n-2}) \no\\
			&+w_{n+1}w_nw_{n-1}(w_{n+2}+w_{n+1}+w_{n-1}+w_{n-2})(2w_n-\kappa) +w_n^2(w_{n+2}w_{n+1}^2+w_{n-1}^2w_{n-2}) \no\\
			&+w_{n+1}w_nw_{n-1}(w_{n+2}+w_{n+1})(w_{n-1}+w_{n-2}) +w_n^2(w_{n+1}^2+w_{n-1}^2)(3w_n-2\kappa-4) \no\\
			&-w_n(z_nw_{n+1}^2+z_{n+1}w_{n-1}^2)+w_{n+1}w_nw_{n-1}(\kappa^2-4w_n\kappa-8w_n+5w_n^2)+w_n^2(w_{n+1}^3+w_{n-1}^3) \no\\
			&+w_n^2(w_{n+1}+w_{n-1})(\kappa^2+8\kappa+3w_n^2-4w_n\kappa-8w_n) -w_n^2[(2z_n+z_{n+1})w_{n+1}+(2z_{n+1}+z_n)w_{n-1}] \no\\
			&+w_n[(2z_{n+1}+z_n(\kappa+2))w_{n+1}+(2z_n+z_{n+1}(\kappa+2))w_{n-1}]+w_n^3[-(z_n+z_{n+1})+\kappa(\kappa+8)] \no\\
			&+w_n^2[(z_n+z_{n+1})(\kappa+2)-4\kappa^2] +w_{n}[z_{n}z_{n+1}-2\kappa(z_n+z_{n+1})]+(C_3+z_n)(C_3-z_{n+1})=0,
	\end{align}
    	where $z_n=C_1+C_2 n,\; C_1=-\frac{4}{a^4},\; C_2=\frac{8}{a^{4}},\;\kappa=0$ and $C_3=\frac{4}{a^{4}}$.
	\end{theorem}
	\begin{proof}
	Using (\ref{s11}) and (\ref{s13}), we have
	\begin{equation}\label{s15}
	R_n=4\left[\beta_{n-1}\beta_n+(\beta_n+\beta_{n+1})^2+\beta_{n+1}\beta_{n+2}\right]-2n-1.
	\end{equation}
	Substituting (\ref{s13}) and (\ref{s15}) into (\ref{s12}), we obtain (\ref{de2}). Equation (\ref{dp342}) follows from the given transformation.
	\end{proof}
\begin{theorem}
	The recurrence coefficient $\bt_n$ satisfies the  differential-difference equation
	$$
	a\beta_{n}^{\prime}(a)=2\beta_{n}\left[2\beta_{n-1}(\beta_{n-2}+\beta_{n-1}+\beta_{n})-2\beta_{n+1}(\beta_{n}+\beta_{n+1}+\beta_{n+2})+1\right].
	$$
	\end{theorem}
	\begin{proof}
	Substituting (\ref{s15}) into (\ref{eq3}) gives the desired result.
	\end{proof}
	\begin{theorem}
	The nontrivial leading coefficient $\mathrm{p}(n,a)$ can be expressed in terms of the recurrence coefficient $\bt_n$ as follows:
	$$
	\mathrm{p}(n,a)=-\frac{1}{2}n a^{2}-2\beta_ n\Big[\beta_{n-1}(\beta_{n-2}+\beta_{n-1}+2\beta_{n}+\beta_{n+1}-a^2)+(\beta_{n}+\beta_{n+1})(\beta_{n}+\beta_{n+1}-a^{2})
	+\beta_{n+1}\beta_{n+2}-\frac{n}{2}-\frac{1}{4}\Big].
    $$
    \end{theorem}
	\begin{proof}
From (\ref{sum}) we have
	\be\label{sumb}
	\sum_{j=0}^{n-1}(\bt_j+\bt_{j+1})=-2\mathrm{p}(n,a)+\bt_n.
    \ee
Inserting the above into     (\ref{s14}), we can express $\mathrm{p}(n,a)$ in terms of $r_n$ and $R_n$:
	\begin{equation}\label{pna1}
		2\mathrm{p}(n,a)=\beta_n+a^2r_n+2\bt_nr_n-\beta_n[R_{n-1}+R_n+4(\beta_{n-1}+\beta_n)(\beta_n+\beta_{n+1})],
	\end{equation}
    Substituting (\ref{s13}) and (\ref{s15}) into (\ref{pna1}), we obtain the desired result.	
	\end{proof}
	\begin{theorem}
    The quantity $\sigma_n(a)=a\frac d{da}\ln\mathbb{P}(n,a)$ can be expressed in terms of the recurrence coefficient $\bt_n$ as follows:
	\begin{align}
	\sigma_n(a) =&\:\left[4\beta_n\left(\bt_{n-1}+\bt_n+\bt_{n+1}\right)-n\right]^2+4a^2(a^2+2\beta_n)\left[4\beta_n\left(\bt_{n-1}+\bt_n+\bt_{n+1}\right)-n\right]\no\\
&-4\beta_n(a^2+\beta_{n-1}+\beta_n)\left[4\beta_{n-1}\beta_n+4(\beta_n+\beta_{n+1})^2+4\beta_{n+1}\beta_{n+2}-2n-1\right]\no\\
&-4\beta_n(a^2+\beta_n+\beta_{n+1})[4\beta_{n-2}\beta_{n-1}+4(\beta_{n-1}+\beta_{n})^2+4\beta_{n}\beta_{n+1}-2n+1].\no
	\end{align}
	\end{theorem}
	\begin{proof}
From (\ref{sig})  and (\ref{s16}) we have
	$$
	\sigma_n(a)=r_n^2+4a^2r_n(a^2+2\beta_n)-4\beta_n\left[R_n(a^2+\beta_{n-1}+\beta_n)+R_{n-1}(a^2+\beta_n+\beta_{n+1})\right],
	$$
where use has been made of (\ref{s12}) to eliminate $\beta_n R_{n-1} R_n$. Substituting  (\ref{s13}) and (\ref{s15}) into the above, we establish the theorem.
	\end{proof}	
	\section{Sextic Freud unitary ensemble $(m=3)$}
	In this case, the weight function $w(x)$ now is
	$$
	w(x;a)=w_0(x)\chi_{(-a,a)^{c}}(x),\qquad x \in \mathbb{R},
	$$
	where $w_0(x)$ is the sextic Freud weight $w_0(x)=\mathrm{e}^{-x^{6}},\; x \in \mathbb{R}.$
	
	Similarly as in the $m=1$ and $m=2$ cases, we find
	\begin{equation}\label{anz3}
	A_n(x)=\frac{R_n(a)}{x^2-a^2}+6 x^4+6 x^2(\beta_n+\beta_{n+1})+ Q_n(a),
	\end{equation}
	\begin{equation}\label{bnz3}
	B_n(x)  =\frac{x\:r_n(a)}{x^2-a^2}+6 x^3\beta_n+q_n(a)x,
	\end{equation}
	where $R_n(a)$ and $r_n(a)$ are given by (\ref{Rnt}) and (\ref{rnt}), and
	\begin{equation}\label{s29}
	Q_n(a):=6\left[\beta_{n-1}\beta_n+(\beta_n+\beta_{n+1})^2+\beta_{n+1}\beta_{n+2}\right],
	\end{equation}
	\begin{equation}\label{s24}
	q_n(a):=6 \beta_n(\beta_{n-1}+\beta_n+\beta_{n+1}).
	\end{equation}

	Substituting (\ref{anz3}) and (\ref{bnz3}) into ($S_1$), we obtain
	\begin{equation}\label{s21}
	r_n+r_{n+1}= R_n.
	\end{equation}
	%\begin{equation}\label{s22}
	%q_n+q_{n+1}=Q_n.
	%\end{equation}
	Similarly, substituting (\ref{anz3}) and (\ref{bnz3}) into ($S_2'$) gives the following five identities:
	\begin{equation}\label{s23}
	a^2	r_n^2=\beta_n R_{n-1} R_n,
	\end{equation}
	%\begin{equation}\label{s24}
	%q_n=6 \beta_n(\beta_{n-1}+\beta_n+\beta_{n+1}),
	%\end{equation}
	\begin{equation}\label{s25}
	n+r_n+2q_n\beta_n=\beta_n\left[Q_{n-1}+Q_n+6(\beta_{n-1}+\beta_n)(\beta_n+\beta_{n+1})\right],
	\end{equation}
    \begin{equation}\label{s26}
    q_n^2+6r_n(a^2+2\beta_n)+6\sum_{j=0}^{n-1}(\beta_j+\beta_{j+1})=6\beta_n\left[R_{n-1}+R_n+Q_n(\beta_{n-1}+\beta_n)+Q_{n-1}(\beta_n+\beta_{n+1})\right],
    \end{equation}
    \begin{equation}\label{s27}
    2r_n(3a^4+q_n+6a^2\beta_n)+\sum_{j=0}^{n-1}Q_j=\beta_n\left[Q_{n-1}Q_n+6R_n(a^2+\beta_{n-1}+\beta_n)+6R_{n-1}(a^2+\beta_n+\beta_{n+1})\right],
    \end{equation}
\begin{align}\label{s28}
&\:a^2r_n^2+4a^4r_n(3a^4+q_n+6a^2\beta_n)+2a^2\sum_{j=0}^{n-1}R_j\no\\
&=\beta_n\left[2a^2R_n\left(Q_{n-1}+6a^2(a^2+\beta_{n-1}+\beta_n)\right)+2a^2R_{n-1}\left(Q_n+6a^2(a^2+\beta_n+\beta_{n+1})\right)-R_{n-1}R_n\right].
	\end{align}
%\begin{remark}
%There are another three long identities involving the sums $\sum_{j=0}^{n-1}R_j,\;\sum_{j=0}^{n-1}Q_j,\;\sum_{j=0}^{n-1}(\beta_j+\beta_{j+1})$ from ($S_2'$). We do not write them down since they will not be used in the following analysis.
%\end{remark}
With these ingredients in hand, we are now ready to derive the following results.
	\begin{theorem}
	The recurrence coefficient $\bt_n$ satisfies the sixth-order nonlinear difference equation
\begin{align}\label{de3}
	&a^2\left[6\beta_n\left(\beta_{n-2}\beta_{n-1}+(\beta_n+\beta_{n+1})^2+\beta_{n-1}(\beta_{n-1}+2\beta_n+\beta_{n+1})+\beta_{n+1}\beta_{n+2}\right)-n\right]^2 \no\\
&=\beta_n\big[6\beta_{n-1}\left(\beta_{n-3}\beta_{n-2}+(\beta_{n-1}+\beta_n)^2+\beta_{n-2}(\beta_{n-2}+2\beta_{n-1}+\beta_n)+\beta_n\beta_{n+1}\right) \no\\
	&+6\beta_n\left(\beta_{n-2}\beta_{n-1}+(\beta_n+\beta_{n+1})^2+\beta_{n-1}(\beta_{n-1}+2\beta_n+\beta_{n+1})+\beta_{n+1}\beta_{n+2}\right)-2n+1\big] \no\\
	&\times\big[6\beta_{n}\left(\beta_{n-2}\beta_{n-1}+(\beta_{n}+\beta_{n+1})^2+\beta_{n-1}(\beta_{n-1}+2\beta_{n}+\beta_{n+1})+\beta_{n+1}\beta_{n+2}\right) \no\\
	&+6\beta_{n+1}\left(\beta_{n-1}\beta_{n}+(\beta_{n+1}+\beta_{n+2})^2+\beta_{n}(\beta_{n}+2\beta_{n+1}+\beta_{n+2})+\beta_{n+2}\beta_{n+3}\right)-2n-1\big].
	\end{align}
	\end{theorem}
	\begin{proof}
	%Substituting (\ref{s24}) into (\ref{s22}), we have
  	%\begin{equation}\label{s29}
  	%Q_n=6\left[\beta_{n-1}\beta_n+(\beta_n+\beta_{n+1})^2+\beta_{n+1}\beta_{n+2}\right].
  	%\end{equation}
  	From (\ref{s25}) and with the aid of (\ref{s29}) and (\ref{s24}), we find
  	\begin{equation}\label{s30}
  	r_n=6\beta_n\left[\beta_{n-2}\beta_{n-1}+(\beta_n+\beta_{n+1})^2+\beta_{n-1}(\beta_{n-1}+2\beta_n+\beta_{n+1})+\beta_{n+1}\beta_{n+2}\right]-n.
  	\end{equation}
  	Substituting (\ref{s30}) into (\ref{s21}), we have
  	\begin{align}\label{s31}
  	R_{n}=&\:6\beta_{n}\left[\beta_{n-2}\beta_{n-1}+(\beta_{n}+\beta_{n+1})^2+\beta_{n-1}(\beta_{n-1}+2\beta_{n}+\beta_{n+1})+\beta_{n+1}\beta_{n+2}\right]\no\\
  	&+6\beta_{n+1}\left[\beta_{n-1}\beta_{n}+(\beta_{n+1}+\beta_{n+2})^2+\beta_{n}(\beta_{n}+2\beta_{n+1}+\beta_{n+2})+\beta_{n+2}\beta_{n+3}\right]-2n-1.
    \end{align}
  	The theorem is established by substituting (\ref{s30}) and (\ref{s31}) into (\ref{s23}).
	\end{proof}
	\begin{remark}
Let $w_n=4\beta_{n}/a^2$. Then, the difference equation	(\ref{de3}) can be transformed into a particular sixth-order equation in the discrete Painlev\'{e} XXXIV hierarchy \cite{Cresswell}. The general form of this equation is very complicated and we decide not to write it down.
	\end{remark}
%\begin{remark}
%For the higher $m>3$ cases, we conjecture that the recurrence coefficient $\bt_n$, with a minor change of variables (might be also $w_n=4\beta_{n}/a^2$), satisfies the $(2m)$th-order equation in the discrete Painlev\'{e} XXXIV hierarchy \cite{Cresswell}.
%\end{remark}
\begin{theorem}
	The recurrence coefficient $\bt_n$ satisfies the  differential-difference equation
	\begin{align} a\beta_{n}^{\prime}(a)=&2\beta_{n}\big[3\beta_{n-1}\left(\beta_{n-3}\beta_{n-2}+(\beta_{n-1}+\beta_{n})^2+\beta_{n-2}(\beta_{n-2}+2\beta_{n-1}+\beta_{n})+\beta_{n}\beta_{n+1}\right)
\no\\
&-3\beta_{n+1}\left(\beta_{n-1}\beta_{n}+(\beta_{n+1}+\beta_{n+2})^2+\beta_{n}(\beta_{n}+2\beta_{n+1}+\beta_{n+2})+\beta_{n+2}\beta_{n+3}\right)+1\big].\no
	\end{align}
	\end{theorem}
	\begin{proof}
	Substituting (\ref{s31}) into (\ref{eq3}) gives the desired result.
	\end{proof}
	\begin{theorem}
	The nontrivial leading coefficient $\mathrm{p}(n,a)$ can be expressed in terms of the recurrence coefficient $\bt_n$ as follows:
	$$
	\mathrm{p}(n,a)=\frac{1}{12}q_n^2+\frac{1}{2}a^2 r_n-\frac{1}{2}\beta_n\left[r_{n-1}+r_{n+1}+Q_n(\beta_{n-1}+\beta_n)+Q_{n-1}(\beta_n+\beta_{n+1})-1\right],
	$$
	where $q_n$, $Q_n$ and $r_n$ are given by (\ref{s24}), (\ref{s29}) and (\ref{s30}), respectively.
    \end{theorem}
	\begin{proof}
Inserting (\ref{sumb}) into     (\ref{s26}) and using (\ref{s21}), we establish the theorem.
	\end{proof}
	\begin{theorem}
    The quantity $\sigma_n(a)=a\frac d{da}\ln\mathbb{P}(n,a)$ can be expressed in terms of the recurrence coefficient $\bt_n$ as follows:
	\begin{align}
	\sigma_{n}(a)=&\:r_n^2+2a^2r_n(3a^4+q_n+6a^2\beta_n)
	-\beta_nR_n\left[Q_{n-1}+6a^2(a^2+\beta_{n-1}+\beta_n)\right]\no\\
&-\beta_nR_{n-1}\left[Q_n+6a^2(a^2+\beta_n+\beta_{n+1})\right],\no
	\end{align}
	where $q_n$, $Q_n$, $r_n$ and $R_n$ are given by (\ref{s24}), (\ref{s29}), (\ref{s30}) and (\ref{s31}), respectively.
	\end{theorem}
	\begin{proof}
Using (\ref{sig}) and (\ref{s28}), and with the aid of (\ref{s23}), we obtain the desired result.
	\end{proof}	
In the end, we would like to point out that one may also study the large $n$ asymptotic behavior of the recurrence coefficient $\bt_n(a)$, the nontrivial leading coefficient $\mathrm{p}(n,a)$, the Hankel determinant $D_n(a)$ and the gap probability $\mathbb{P}(n,a)$ as in \cite{MW}. We will leave this problem to a future investigation.

\section{Discussion}
In this paper, by considering the gap probability problem of the Freud unitary ensemble, we have established the relationship between the recurrence coefficients of orthogonal polynomials with respect to the weight (\ref{weight}) and the discrete Painlev\'{e} XXXIV hierarchy for the $m=1, 2$ and $3$ cases. For the higher $m>3$ cases, we conjecture that the recurrence coefficient $\bt_n$, with a minor change of variables (might be also $w_n=4\beta_{n}/a^2$), satisfies the $(2m)$th-order equation in the discrete Painlev\'{e} XXXIV hierarchy \cite{Cresswell}. It was also pointed out in \cite{Cresswell} that the discrete Painlev\'{e} XXXIV hierarchy is linked to the discrete Painlev\'{e} II hierarchy by Miura transformations.
%It is interesting to consider the relationship between the recurrence coefficients of various semi-classical orthogonal polynomials (especially the symmetric case with even weights, e.g., \cite{Min2021}) and the discrete Painlev\'{e} equations.
	
	\section*{Declaration of competing interest}
The authors declare that they have no known competing financial
interests or personal relationships that could have appeared to influence
the work reported in this paper.
	%The authors have no competing interests to declare that are relevant to the content of this article.
	%This work does not have any conflict of interest.
	%No data was used for the research described in the article.
	% Acknowledgements
	\section*{Acknowledgements}
	This work was partially supported by the National Natural Science Foundation of China under grant number 12001212, by the Fundamental Research Funds for the Central Universities under grant number ZQN-902 and by the Scientific Research Funds of Huaqiao University under grant number 17BS402.
\section*{Data availability}
No data was used for the research described in the article.
	
\end{document}